\documentclass[3p, preprint, 11pt]{elsarticle}

\usepackage{graphicx}
\usepackage{amsfonts}
\usepackage{amsmath}
\usepackage{amssymb}
\usepackage{amsthm}
\usepackage{xcolor}
\definecolor{lightblue}{rgb}{ 0.75, 0.75, 1.0 }
\definecolor{black}{rgb}{ 0.0, 0.0, 0.0 }
\usepackage{enumerate}
\usepackage{algorithm}
\usepackage{algorithmic}

\newtheorem{theorem}{Theorem}

\newtheorem{corollary}{Corollary}[section]

\newtheorem{hypothesis}[theorem]{Hypothesis}

\newtheorem{problem}{Problem}[section]

\DeclareMathOperator{\Ti}{Ti}
\DeclareMathOperator{\Ei}{Ei}
\DeclareMathOperator{\RePart}{Re}

\renewcommand{\algorithmicrequire}{\textbf{Input:}}
\renewcommand{\algorithmicensure}{\textbf{Output:}}

\begin{document}

\begin{frontmatter}

\title{Inconsistency in the analytic hierarchy process under data uncertainties}

\author[a1]{Przemys\l{}aw Rutka\corref{cor1}}
\ead{przemyslaw.rutka@mimuw.edu.pl}

\cortext[cor1]{Corresponding author}
\affiliation[a1]{organization={Institute of Informatics, Faculty of Mathematics,
  Informatics, and Mechanics, University of Warsaw},
            addressline={ul.~Banacha~2}, 
            city={Warsaw},
            postcode={02-097}, 
            country={Poland}}
            
\author[a2]{Ryszard Smarzewski}
\ead{ryszard.smarzewski@wat.edu.pl}

\affiliation[a2]{organization={Institute of Mathematics and Cryptology,
  Faculty of Cybernetics, Military University of Technology},
             addressline={ul.~gen.~S.~Kaliskiego~2},
             city={Warsaw},
             postcode={00-908},
             country={Poland}}

\begin{abstract}
 In this paper, we study the inconsistency of matrices in multicriteria group decision making under the assumption of uncertainty in their entries. The study binds together an 
 $n$-dimensional probabilistic product measure $P$ and the
 Koczkodaj indicator $K_{n}$ of order $n$ used in pairwise
 comparisons. In the case of $n$ greater than $3$, it starts
 with a new definition of $K_{n}$ based on the Steiner
 triple system of order $n$. It is of interest that the new
 indicator requires only $O \left( n ^{2} \right)$ evaluations,
 while the old one $O \left( n ^{3} \right)$. Our study results
 in analytic formulae for the probabilities of events 
 $K_{n} \leq \varepsilon$ corresponding to uniform, exponential,
 and normal distribution functions, whenever the remainder of
 $n$ divided by 6 equals 1 or 3. For the remaining values of
 $n$, we present analytic formulae for their strict bounds. 
 The formulae lead to much more efficient algorithms when compared with Monte Carlo simulations. This unexpected feature is illustrated by a few numerical experiments.
\end{abstract}


\end{frontmatter}	

\section{Introduction and preliminaries} \label{SEC_INTRO}
Pairwise comparisons is the most important and fundamental method in multicriteria group decision making. It emerged in the year 1977 as an essential part of the Analytic Hierarchy Process (AHP), attributed to \cite{Saat1s}. 
The aim of this process is to derive a vector of priorities
ranking the alternatives under consideration by importance.
These priorities are obtained from relative preferences made
by human experts (or other actors like responders, knowledge
suppliers, voters) who compare each pair of alternatives and
subjectively tell how many times one alternative is more
important than the second. The numerical assessments given by
an individual expert are collected in a square PC matrix of 
order equal to the number of all alternatives.

The ideal situation in pairwise comparisons occurs when an
expert makes perfectly coherent judgements about the alternatives. Then a PC matrix of assessments is said to be consistent and can be represented in terms of a priority vector whose length equals the number of alternatives. However, most human assessments deviate from such perfect consistency, especially when they are made for many alternatives. Therefore, during a decision making process, the consistency of every provided PC matrix has to be verified, and those that are not sufficiently consistent, i.e., meaningless and unreliable for a decision maker, should be repaired by the expert or excluded from the process.

The most popular methods for measuring the inconsistency of a
PC matrix use indicators proposed by \cite{Saat1s} and \cite{Kocz4s}. These are eigenvalue- and distance-based indicators, respectively. The eigenvalue-based indicator was criticized by \cite{Kocz5s}, who
argued that its soundness must be questioned. They 
substantiated this claim using well-crafted counterexamples.

According to \cite{Kocz5s}, the long
history of measuring inconsistency in pairwise comparisons
includes important contributions by \cite{Kend1}, \cite{Saat1s},
\cite{Laar1s}, \cite{Kula1s}, and \cite{Szyb1s}.
Moreover, it is summarized by \cite{Brun1s,Brun3s,Brun4s} and \cite{Mazu1}. For such measurement, several concepts of inconsistency indicators have been introduced and applied in practice; cf. \cite{Agua1s},  \cite{Brun3s}, \cite{Kocz4s}, \cite{Kocz2s}, \cite{Kocz5s},  \cite{Liu1s}, \cite{Magn1,Magn2}, \cite{Pela1}, \cite{Saat1s}, \cite{Stei1s},
and the authors \citep{Smar3s}. In this regard, we mention a
hard unsolved problem.
\begin{problem} \label{PRO_PLAINESTIMATECHANCE}
  For a given inconsistent matrix, estimate the probability of generating a matrix that is as close as possible to the given one while maintaining a probabilistic inconsistency level below a fixed threshold.
\end{problem}
To solve this problem, it is necessary to choose an
appropriate probability measure and compute the probability that the inconsistency value of a given indicator falls below a specified threshold. We note that knowledge of 
chance in Problem \ref{PRO_PLAINESTIMATECHANCE} is of importance
both from a computational and theoretical point of view.
For example, it helps determine the appropriate number of considered alternatives and the threshold when computing accurate approximations with an admissible degree of inconsistency. The main objective of our study is to solve Problem \ref{PRO_PLAINESTIMATECHANCE} for some classical distribution functions.

The paper is organized as follows: In Sections 
\ref{SEC_UNIFORM3D} - \ref{SEC_NORMAL3D} we
investigate how the perturbations of a given PC matrix of order
$3$ have an effect on the generated consistent matrix under uniform, exponential, and normal distributions. In the first
case, we additionally establish how the rate of inconsistency
depends on the support of the corresponding density function.
The effective analytic formulae established in these sections are compared with the Monte Carlo simulations of an approximate consistent matrix, which require an enormous 10 million samples to achieve results comparable to ours. In Section \ref{SEC_PROBABILITYORDERn} we consider pairwise comparisons for $n$ alternatives, where $n > 3$. We begin with a new definition of a Koczkodaj-type indicator. The new indicator, which is based on the Steiner triple system of order $n$ \citep{Fort1s, Kirk1s, Lind1}, has an optimal complexity of $O \left( n ^{2} \right)$, whereas the Koczkodaj indicator is $O \left( n ^{3} \right)$. Subsequently, we derive explicit formulae for the inconsistency probabilities for $n \equiv 1 \left(\text{mod } 6\right)$ and $n \equiv 3 \left(\text{mod } 6\right)$. In the remaining cases strict estimates are presented. Finally, in the Conclusions, we propose a divide-and-conquer algorithm for group decision making that adjusts the inconsistency threshold based on the number of experts.

\section{Indicator measure of inconsistency} \label{SEC_PRELIMINARIES}
The purpose of multicriteria group decision making is to derive a priority vector that ranks the alternatives under consideration by importance. Priorities are obtained from relative preferences expressed by human experts who compare each pair of alternatives and subjectively tell how many times one alternative is more important than the second. These numerical assessments are collected in a PC $n \times n$ matrix of the form
\begin{equation*}
  M =
  \begin{bmatrix}
    1 & m_{12} & \cdots  & m_{1n} \\
    m_{21} & 1 & \cdots  & m_{2n} \\
    \vdots  & \vdots  & \ddots  & \vdots  \\
    m_{n1} & m_{n2} & \cdots  & 1
  \end{bmatrix},
\end{equation*}
where $n>2$ is the number of alternatives considered and
\begin{equation*}
  m_{ij} = \frac{1}{m_{ji}} > 0 \quad \text{for all}
  \quad i, j.
\end{equation*}
A PC matrix $M$ is said to be consistent if
\begin{equation*}
  m_{ij} m_{jk} = m_{ik} ,
\end{equation*}
for all $\binom{n}{3}$ triads from the index set
\begin{equation*}
  \mathcal{T} \left( n\right) = \left\{ \left( i, j, k \right) 
  : 1 \leq i < j < k \leq n \right\} .
\end{equation*}
One of the most popular methods of measuring the
inconsistency of $M$ is based on the following definition of an indicator \citep{Dusz1s}:
\begin{equation*}
  K_{n} \left( M \right) = \max_{\left( i, j, k\right) \in
  \mathcal{T} \left( n \right)} K \left( M_{ijk} \right) ,
\end{equation*}
where $K$ is the Koczkodaj indicator of order $3$ \citep{Kocz4s}:
\begin{equation*}
  K \left( M_{ijk} \right) = 1 - \min \left\{ \frac{ m_{ik} }{ m_{ij}m_{jk} },
   \frac{ m_{ij}m_{jk} }{ m_{ik} } \right\}
\end{equation*}
for principal submatrices $M_{ijk}$ of $M$ having the form
\begin{equation*}
    M_{ijk} = \left[
    \begin{array}{lll}
      1 & m_{ij} & m_{ik} \\
      m_{ji} & 1 & m_{jk} \\
      m_{ki} & m_{kj} & 1
    \end{array}    
    \right] .
\end{equation*}
We note that $K_{3} = K$, $0 \leq K_{n} \leq 1$ and
$K_{n} \left( M \right) = 0$ if and only if $M$ is consistent.
On the other hand, the definition of $K_{n} \left( M \right)$
has a serious drawback for $n > 3$, because the following
hypothesis is false for 
$\mathcal{S} \left( n \right) = \mathcal{T} \left( n \right)$.
This problem will be studied in Section \ref{SEC_PROBABILITYORDERn}.
\begin{hypothesis} \label{HYP_NONOVERLAPPING}
    There exists a subset $\mathcal{S} \left( n \right)$ of
    $\mathcal{T} \left( n \right)$ such that every entry
    $m_{rs}$ of $M$ above the main diagonal is an entry of
    exactly one submatrix $M_{ijk}$ for a triad
    $\left( i, j, k \right) \in \mathcal{S} \left( n \right)$. 
\end{hypothesis}
Using this notation, we can write the problem stated in Section \ref{SEC_INTRO} as follows.
\begin{problem} \label{PRO_ESTIMATECHANCE}
  For a given inconsistent $n \times n$ PC matrix $M$, estimate the chance of randomly generating a matrix $A$ (from the set of $n \times n$ PC matrices) that is as close as possible to $M$ and satisfies the inequality
  \begin{equation*}
      K_{n} \left( A \right) \leq \varepsilon, \quad
      \varepsilon > 0.
  \end{equation*}
\end{problem}
The main objective of our study is to solve Problem \ref{PRO_ESTIMATECHANCE} for some distribution functions.
For this purpose, we use the Lebesgue-Stieltjes probabilistic product measures defined on the $\binom{n}{2}$-fold Cartesian product of Borel sets in $\mathbb{R}$. Moreover, we consider $K_{n} \left( M \right)$ as a function of $\binom{n}{2}$ positive variables $m_{ij}$ $\left( 1 \leq i < j \leq n \right)$.

\section{Uniform perturbation of $3 \times 3$ matrices}
\label{SEC_UNIFORM3D}
Let $F$ and $f$ denote the uniform distribution and density in
$\mathbb{R} ^{n}$ defined by
\begin{equation*}
  F \left( x, z, y\right) =
    F\left( x \right) F\left( z \right) F\left( y \right) 
  \quad \text{and} \quad
  f \left( x, z, y\right) =
    f\left( x \right) f\left( z \right) f\left( y \right),  
\end{equation*}
where the one-dimensional functions are equal to
\begin{equation*}
  F \left( t \right) =\max \left\{ 0, \min \left( 1, t \right)  
    \right\}
  \quad \text{and} \quad
  f \left( t \right) = \left\{ 
    \begin{array}{ll}
      1, & 0< t < 1, \\ 
      0, & \text{otherwise.}
    \end{array}
  \right.
\end{equation*}
Then the Lebesgue-Stieltjes measure $P = dF$ is concentrated on
\begin{equation*}
  Q = \left\{ \left( x, z, y\right) \in \mathbb{R} ^{3}: 
    0 < x, z, y \leq 1 \right\} .
\end{equation*}
In the following, we write a $3 \times 3$ PC matrix in the form
\begin{equation*}
  M = \left[
    \begin{array}{ccc}
      1 & x & y \\
      1 / x & 1 & z \\
      1 / y & 1 / z & 1
    \end{array}
  \right] .
\end{equation*}
Then the indicator
\begin{equation*}
  K \left( M \right) = K \left( x, y, z \right) = 1
  - \min \left\{ \frac{y}{xz}, \frac{xz}{y} \right\}
\end{equation*}
is a function defined everywhere on $\mathbb{R} ^{3}$, except
for the points of planes $x = 0$, $y = 0$, and $z = 0$. It has a
continuous extension to all points of these planes, except those lying on the coordinate axes. Indeed, we have
\begin{equation*}
  K \left( x, y, z\right) = 1 - \frac{1}{2} \left( \frac{y}{xz}
    + \frac{xz}{y} - \left\lvert \frac{y}{xz} - \frac{xz}{y} \right\rvert \right), 
\end{equation*}
and so $K \left( x, y, z \right) \to 0$ as $z > 0$, $y > 0$, and
$x$ decreases to $0$. Hence, we can put
$K \left( 0, z, y \right) = 1$ and apply symmetry to prove the
claim. Now we establish the first of two main results in this section.

\begin{theorem} \label{THR_PKUNITUNIFORM}
  If $0 < \varepsilon < 1$, then we have
  \begin{equation*}
    P \left( K \leq \varepsilon \right) = \varepsilon + 
      \frac{1 - \varepsilon}{2} \ln \left( 1 - \varepsilon \right) .
  \end{equation*}
\end{theorem}

\begin{proof}
    By the definition of indicator, it follows that
    \begin{equation*}
      K = \left\{ 
      \begin{array}{ll}
        1 - \frac{y}{xz}, & y \leq xz, \\[+2pt] 
        1 - \frac{xz}{y}, & y > xz.
      \end{array}
      \right. 
    \end{equation*}
    Hence, the inequality $K \leq 1 - \varepsilon$ is satisfied in
    cube $Q$ if and only if
    \begin{equation*}
      \varepsilon xz \leq y \leq \min \left\{ 1, 
        \frac{xz}{\varepsilon} \right\},\quad
      \left( x, z, y\right) \in Q.
    \end{equation*}
    The volume of all points $\left( x, y, z \right)$ that satisfy
    this inequality is equal to
    \begin{align*}
      P \left( K \leq 1 - \varepsilon \right)
        & = \int_{0} ^{1} \int_{0} ^{1} \left[ \min
          \left\{ 1, \frac{xz}{\varepsilon} \right\} 
          - \varepsilon xz \right] dzdx \\
        & = \frac{1}{2} \left[ 1 + \frac{1}{4}
          \left( \frac{1}{\varepsilon} - 2\varepsilon \right)    
          - \int_{0} ^{1} \int_{0} ^{1} \left\lvert 1 
          - \frac{xz}{\varepsilon} \right\rvert dzdx
      \right].
    \end{align*}
    In order to compute the last integral, we use the formula
    \begin{equation*}
      \min \left\{ s, t \right\} = 
        \frac{1}{2} \left( s + t - \left\lvert s - t\right\rvert
        \right).
    \end{equation*}
    Then it can be evaluated as follows
    \begin{align*}
      \int_{0} ^{\varepsilon} \int_{0} ^{1} \left( 1 - 
        \frac{xz}{\varepsilon} \right) dzdx
      & + \int_{\varepsilon} ^{1} \int_{0} ^{\varepsilon / x} 
        \left( 1 - \frac{xz}{\varepsilon} \right) dzdx 
        + \int_{\varepsilon} ^{1} \int_{\varepsilon / x} ^{1}
        \left( -1 + \frac{xz}{\varepsilon} \right) dzdx \\
      & = \frac{3 \varepsilon}{2} - \varepsilon \ln \varepsilon -1 + \frac{1}{4 \varepsilon}.
    \end{align*}
    Hence, we eventually get
    \begin{equation*}
      P \left( K \leq 1 - \varepsilon \right) = 1 - \varepsilon 
        + \frac{\varepsilon}{2} \ln \left( \varepsilon \right) .
    \end{equation*}
\end{proof}

Note that the probability
\begin{equation*}
  P \left( K \leq \varepsilon \right) = \varepsilon + 
    \frac{1 - \varepsilon}{2} \ln \left( 1 - \varepsilon \right), \quad 0 < \varepsilon < 1,
\end{equation*}
converges to the expected values $0$ and $1$ for 
$\varepsilon = 0$ and $\varepsilon = 1$. The orders of this
convergence are given in the following corollary.

\begin{corollary}
    We have
    \begin{align*}
      \text{a)}\quad P \left( K \leq \varepsilon \right) & = 
        o \left( \varepsilon \right), \quad \varepsilon \to 0,\\
      \text{b)}\quad P \left( K \leq \varepsilon \right) & =
        o \left[ \left( 1 - \varepsilon \right) 
        \ln \left( 1 - \varepsilon \right) \right], \quad
        \varepsilon \to 1.
    \end{align*}    
    \label{COR_UNIFORMPROBABILITYP}
\end{corollary}

\begin{proof}
  Since
  \begin{equation*}
    \lim_{\varepsilon \to 0} \frac{P \left( K \leq \varepsilon
    \right)}{\varepsilon} = \lim_{\varepsilon \to 0}
    \left[ 1 - \frac{1 + \ln \left( 1 - \varepsilon \right)}{2}
    \right] = \frac{1}{2}
 \end{equation*}
 and
 \begin{equation*}
    \lim_{\varepsilon \to 1} \frac{P \left( K\leq \varepsilon
    \right) - 1}{\left( 1 - \varepsilon \right) 
    \ln \left( 1 - \varepsilon \right) }
    = \lim_{\varepsilon \to 1} \left[ \frac{1}{2} + 
    \frac{1}{\ln \left( 1 - \varepsilon \right) }\right] 
    = \frac{1}{2}
  \end{equation*}
  are non-zero, the proof is complete.
\end{proof}

Usually, it is necessary to sample the PC matrix $M$ that depends on the variables
$\left( x, y, z \right)$ of a larger cube $Q_{a} = a Q$ than the
unit cube $Q$. Then we should take $P_{a} = d F_{a}$, where
\begin{equation*}
  F_{a} \left( t \right) = \max \left\{ 0, \min \left\{ 1, \frac{t}{a} \right\} \right\} ,
    \quad a > 1.
\end{equation*}

\begin{theorem} \label{THR_UNIFORMPROBABILITYPa}
  If $0 < \varepsilon = 1 - \tau < 1 < a$, then
  \begin{align*}
    \text{a)} \quad P_{a} \left( K \leq \varepsilon \right) 
      & = \frac{3}{4a \tau} - \frac{3 \tau}{4a} 
      + \frac{\tau}{2a} \ln \frac{\tau}{a} 
      + \frac{\ln \left( a \tau \right)}{2a \tau},\quad
      \text{if} \quad a \tau > 1, \\
    \text{b)} \quad P_{a} \left( K \leq \varepsilon \right) 
      & = 1 -\frac{3\tau}{4a} - \frac{a\tau}{4}
      + \frac{\tau}{2a} \ln \frac{\tau}{a}, \quad
      \text{if} \quad a \tau \leq 1.
  \end{align*}
\end{theorem}

\begin{proof}
  Inequality $K \leq 1 - \varepsilon$ is satisfied in $Q_{a}$ if
  and only if 
  \begin{equation*}
    \min \left\{ a, \varepsilon xz \right\} \leq y
    \leq \min \left\{ a, \frac{xz}{\varepsilon} \right\} ,\quad \left( x,z,y\right) \in Q_{a}.
  \end{equation*}
  Hence, one can express the following probability in integral
 form:
  \begin{equation*}
    P_{a} \left( K \leq 1 - \varepsilon \right) = 
      \frac{1}{a ^{3}} \int_{0} ^{a} \int_{0} ^{a}
      \left[ \min \left\{ a, \frac{xz}{\varepsilon} \right\} - \min \left\{ a, \varepsilon xz \right\} \right] dzdx.
  \end{equation*}
  By changing variables $z = av$ and $x = au$, we obtain
  \begin{align*}
  P_{a} \left( K \leq 1 - \varepsilon \right) 
    & = \int_{0} ^{1} \int_{0} ^{1} \left[ 
      \min \left\{ 1, \frac{auv}{\varepsilon} \right\} - \frac{\varepsilon uv}{a} \right] dvdu \\
    & \quad - \int_{0} ^{1} \int_{0} ^{1} \left[ 
      \min \left\{ 1, a \varepsilon uv \right\} 
      - \frac{uv}{a \varepsilon} + \frac{uv}{a \varepsilon}
      - \frac{\varepsilon uv}{a} \right] dvdu \\
    & = \mathcal{V} \left( \frac{\varepsilon}{a} \right) 
      - \mathcal{V} \left( \frac{1}{a \varepsilon} \right) 
      - \frac{1}{4} \left( \frac{1}{a \varepsilon} - \frac{\varepsilon}{a} \right), 
  \end{align*}
  where
  \begin{align*}
    \mathcal{V} \left( t \right)
    & = \int_{0} ^{1}\int_{0} ^{1}\left[ \min
      \left\{ 1, \frac{xz}{t} \right\} - txz\right] dzdx \\
    & = \left\{ 
      \begin{array}{ll}
        1 - t + \frac{1}{2} t \ln t, & 0 < t < 1, \\
        \frac{1}{4t} - \frac{t}{4}, & t \geq 1.
      \end{array}
    \right.
  \end{align*}
  Indeed, the first formula has been computed in the proof of Theorem \ref{THR_PKUNITUNIFORM} and the second is a direct consequence of identity
  $\min \left\{ 1, \frac{xz}{t} \right\} = \frac{xz}{t}$,
  $t \geq 1$. Finally, if $a \varepsilon > 1$, then
  \begin{equation*}
    P_{a} \left( K \leq 1 - \varepsilon \right) 
      = \frac{3}{4a \varepsilon} - \frac{3 \varepsilon}{4a} 
      + \frac{\varepsilon}{2a} \ln \frac{\varepsilon}{a} 
      + \frac{\ln \left( a \varepsilon \right)}{2a \varepsilon},
  \end{equation*}
  and if $a \varepsilon \leq 1$, then
  \begin{equation*}
    P_{a} \left( K \leq 1 - \varepsilon \right)
     = 1 - \frac{3 \varepsilon}{4a} - \frac{a \varepsilon}{4} 
     + \frac{\varepsilon}{2a} \ln \frac{\varepsilon}{a}.
  \end{equation*}
\end{proof}

The formulae in Theorems \ref{THR_PKUNITUNIFORM} and
\ref{THR_UNIFORMPROBABILITYPa} admit efficient computation of
the probabilities $P \left( K \leq \varepsilon \right)$ and 
$P_{a} \left( K \leq \varepsilon \right)$ of getting by chance 
a PC matrix $M$ having the inconsistency indicator
$K \left( M \right)$ not greater than $\varepsilon$. These
formulae yield the following algorithm for computing the
probability $P_{a}$.

\begin{algorithm}
\caption{\small Computing the uniform probability of 
inconsistency $P_{a} \left( K \leq \varepsilon \right)$.}\label{ALG_PROBABILITYUNIFORM}
\begin{algorithmic}[1]
{\small
  \REQUIRE the level of inconsistency $\varepsilon \in \left( 0, 1 \right)$ and the scale $a$.
  \ENSURE the probability $P_{a} \left( K \leq \varepsilon \right)$.
  \STATE $\tau \leftarrow 1 - \varepsilon$
  \STATE $s \leftarrow - \frac{3\tau}{4a} + \frac{\tau}{2a}\ln\frac{\tau}{a}$
  \IF{$a\tau \leq 1$}
    \STATE $P \leftarrow 1  - \frac{a\tau}{4} + s$
  \ELSE
    \STATE $P \leftarrow \frac{3}{4a\tau} + s + \frac{\ln\left(a\tau\right)}{2a\tau}$
  \ENDIF \\
  \RETURN $P$
}
\end{algorithmic}
\end{algorithm} \medskip

\noindent {\bf Numerical experiments}\\
In the experiments, we used Algorithm \ref{ALG_PROBABILITYUNIFORM} to compute the probability values $P_{a} \left( K \leq \varepsilon \right)$ for the parameters 
\begin{equation*}
  a = 1, 2, \ldots, 10 \quad \text{and} \quad 
  \varepsilon = 0.05,\ 0.1,\ 0.2,\ \ldots\ 0.9,\ 0.95.
\end{equation*}
The numerical results were calculated with a precision of 4
decimal places. Figure \ref{FIG_UNIFORM} shows a plot of the probabilities $P_{a} \left( K \leq \varepsilon \right)$ 
generated for some values of given parameters. Each probability
we obtained was compared with that from a classical Monte
Carlo simulation \citep{Herm1s, Szyb1s}. For this purpose, we computed $N = 10 ^{7}$ random PC matrices $M$ by choosing random numbers $x, y, z$ each time from the interval $\left( 0, a\right]$, and counted how many times $K \left( M\right) \leq \varepsilon$. The ratio of the latter number and $N$ gave an approximation of the probability $P_{a} \left( K \leq \varepsilon \right) $. The results we derived from the Monte Carlo simulations differ from the previous ones by no more than $0.0004$.

\begin{figure}[ht]
  \centering \includegraphics[width=0.75\textwidth]{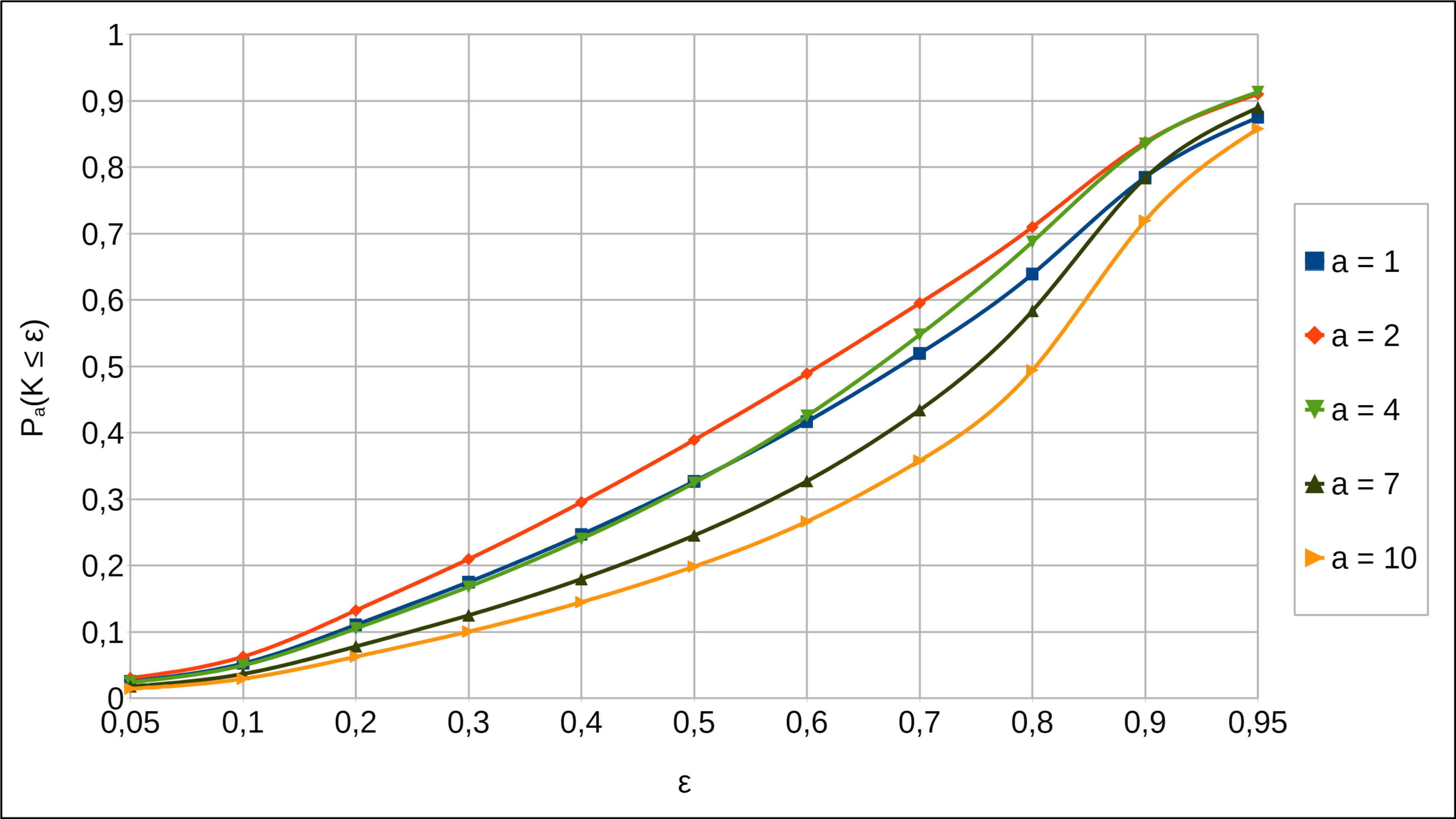}
  \caption{Uniform probabilities $P_{a} \left( K \leq \varepsilon \right) $ generated by Algorithm \ref{ALG_PROBABILITYUNIFORM} for various values of $a$.    
  }\label{FIG_UNIFORM}
\end{figure}

\section{Exponential perturbation of $3 \times 3$ matrices} \label{SEC_EXPONENTIAL3D}
It is of interest that an analytical formula for the probability $P \left( K \leq \varepsilon \right)$ can also be established for other distributions of positive data 
$\left( x, z, y \right)$, which define a PC matrix.
In this section, we consider the probabilistic inconsistency 
$P \left( K \leq \varepsilon \right)$ generated by the Lebesgue-Stieltjes product measure $P = d F$ in $\mathbb{R} ^{3}$ of the following one-dimensional exponential distribution and corresponding density
\begin{equation*}
  F \left( t \right) =\max \left\{ 0, 1 -e^{-\lambda t} \right\} \quad \text{and} \quad
  f \left( t \right) =\left\{ 
  \begin{array}{ll}
    0, & t < 0, \\ 
    \lambda e^{-\lambda t}, & t > 0.
  \end{array}
  \right. 
\end{equation*}
In this case, $P \left( K \leq \varepsilon \right)$ can be
expressed in terms of the exponential integrals 
$E_{1} \left( t \right) = -\Ei \left( -t \right)$ or the 
confluent hypergeometric function
$U \left( 1, 1, t \right) = e ^{t} E_{1} \left( t \right)$. We
recall that the exponential integral $E_{1}$ is defined
according to
\begin{equation*}
  E_{1} \left( t \right) = \int_{1} ^{\infty} \frac{e ^{-tx}}{x}
  dx = -\gamma - \ln t - \sum_{k = 1} ^{\infty} 
  \frac{ \left( -t \right) ^{k}}{k!k},\quad t > 0.
\end{equation*}
Here $\gamma$ is the Euler-Mascheroni constant with the numerical value $0.5772156649$, exact up to $10$ decimal places.

\begin{theorem} \label{THR_EXPPROBABILITY}
  If $0 < \varepsilon = 1 - \tau < 1$ and $\lambda > 0$, then
  \begin{equation*}
    P \left( K \leq \varepsilon \right) = 
      \frac{\lambda}{\tau} e ^{\frac{\lambda}{\tau}}
      E_{1} \left( \frac{\lambda}{\tau} \right) - \lambda \tau e ^{\lambda \tau} E_{1} \left( \lambda \tau \right) . 
  \end{equation*}
\end{theorem}

\begin{proof}
  By the definition of indicator $K$, we conclude that a point $\left( x, z, y \right) $ with positive coordinates satisfies the inequality $ K \left( x, z, y \right) \leq 1 - \varepsilon $ if and only if it belongs to the set
  \begin{equation*}
    D_{\varepsilon} = \left\{ \left( x, z, y\right) \in 
    \mathbb{R}^{3} : \varepsilon x z \leq y < 
    \frac{xz}{\varepsilon},\ x > 0,\ z > 0 \right\} .
  \end{equation*}
  Since we have $\int_{\mathbb{R}^{3}} dF = 1$, it follows that 
  \begin{align*}
    P \left( K \leq 1 - \varepsilon \right)
      & = \int_{0} ^{\infty} \int_{0} ^{\infty} 
        \int_{\varepsilon xz}^{ xz / \varepsilon} 
        \lambda e ^{-\lambda y} dy dF \left( z \right) 
        dF \left( x \right)  \\
      & = \int_{0} ^{\infty} \int_{0} ^{\infty}
        \left[ -\lambda e ^{-\lambda z \left( 1 + 
        \frac{x}{\varepsilon} \right)} + \lambda 
        e ^{-\lambda z \left( 1 + \varepsilon x \right)} \right] dz dF \left( x \right) \\
      & = \lambda \int_{0} ^{\infty} \left( 
        \frac{e ^{-\lambda x}}{1 + \varepsilon x} 
        - \frac{e ^{-\lambda x}}{1 + \frac{x}{\varepsilon}}
        \right) dx \\
      & = \frac{\lambda e ^{\frac{\lambda}{\varepsilon  }}}  
        {\varepsilon } \int_{1} ^{\infty}
        \frac{e ^{-\frac{\lambda x}{\varepsilon}}}{x} dx
        - \lambda \varepsilon e ^{\lambda \varepsilon}
        \int_{1} ^{\infty} \frac{e ^{-\lambda \varepsilon x}}
        {x}dx.
  \end{align*}
  Hence, the definition of the exponential integral $E_{1}$
  yields
  \begin{equation*}
    P \left( K \leq 1 - \varepsilon \right) 
      = -\lambda \varepsilon e ^{\lambda \varepsilon} 
      E_{1} \left( \lambda \varepsilon \right) +
      \frac{\lambda e ^{\frac{\lambda}{\varepsilon}}}
      {\varepsilon} E_{1} \left( \frac{\lambda }{\varepsilon} 
      \right),
  \end{equation*}
  which completes the proof.
\end{proof}

The probability in Theorem \ref{THR_EXPPROBABILITY} can be
rewritten in the form
\begin{equation}
  P \left( K \leq \varepsilon \right)
    = \lambda \int_{0} ^{\infty} g \left( x \right) e ^{-x}dx,
  \label{EQ_EXPONENTIALPROBABILITY}
\end{equation}
where
\begin{equation*}
  g \left( x \right) = \frac{1}{\lambda + 
  \left( 1 - \varepsilon \right) x} 
  - \frac{1 - \varepsilon}{\lambda 
  \left( 1 - \varepsilon \right) + x}.
\end{equation*}
Hence, it can be efficiently approximated by the Gauss-Laguerre
quadrature formula
\begin{equation*}
  P \left( K \leq \varepsilon \right) \approx \lambda 
  \sum_{k = 1}^{n} w_{k} g \left( x_{k} \right),
\end{equation*}
where $w_{k}$ are weights and $x_{k}$ nodes of the Gauss-Laguerre quadrature of order $n$ \citep{Gaut1,Niki1}.  
The weights and nodes can be easily precomputed using
the divide-and-conquer $O \left( n ^{2} \right)$-algorithm due to \cite{Gu1s} or by the iterative $O \left( n \right)$-algorithm of \cite{Glas1s}.
It is worth noting that the values of
$P \left( K \leq \varepsilon \right)$ can also be computed
by applying the MATLAB function
$\verb|expint(t)| \approx E_{1} \left( t \right)$. In this case,
we use Theorem \ref{THR_EXPPROBABILITY} to obtain a MATLAB command of the form: 
\begin{equation*}
  P \left( K \leq \varepsilon \right) \approx
    \verb|x * exp(x) * expint(x) - y * exp(y) * expint(y)|,
\end{equation*}
where
\begin{equation*}
  \verb|x| = \frac{\lambda}{1 - \varepsilon}
  \quad \text{and} \quad
  \verb|y| = \lambda \left( 1 - \varepsilon \right).
\end{equation*}

\noindent {\bf Numerical experiments} \\
As part of the experiments, we computed the probability
values $P$ (rounded to 4 decimal places) for the parameters
\begin{equation*}
  \varepsilon = 0.05,\ 0.1,\ 0.2,\ 0.3,\ 0.4,\ 0,5,\ 0.6,\ 0.7,
  \ 0.9,\ 0.95
\end{equation*}
and
\begin{equation*}
  \lambda = 0.1,\ 0.4,\ 0.7,\ 1,\ 1.5,\ 2,\ 2.5,\ 3,
  \ 3.5,\ 4,5
\end{equation*}
using the MATLAB command. Some selected results are presented
in Figure \ref{FIG_LAGUERRE}. For completeness, we compared the results with the outputs of a Monte Carlo simulation similar to
the one in Section \ref{SEC_UNIFORM3D}. 
Every 10 million random PC matrices generated according to an exponential distribution with parameter $\lambda > 0$ yielded a single approximate probability after verifying their inconsistency level. The maximum difference between the corresponding probabilities did not exceed $0.0003$.

\begin{figure}[!ht]
  \centering \includegraphics[width=0.75\textwidth]{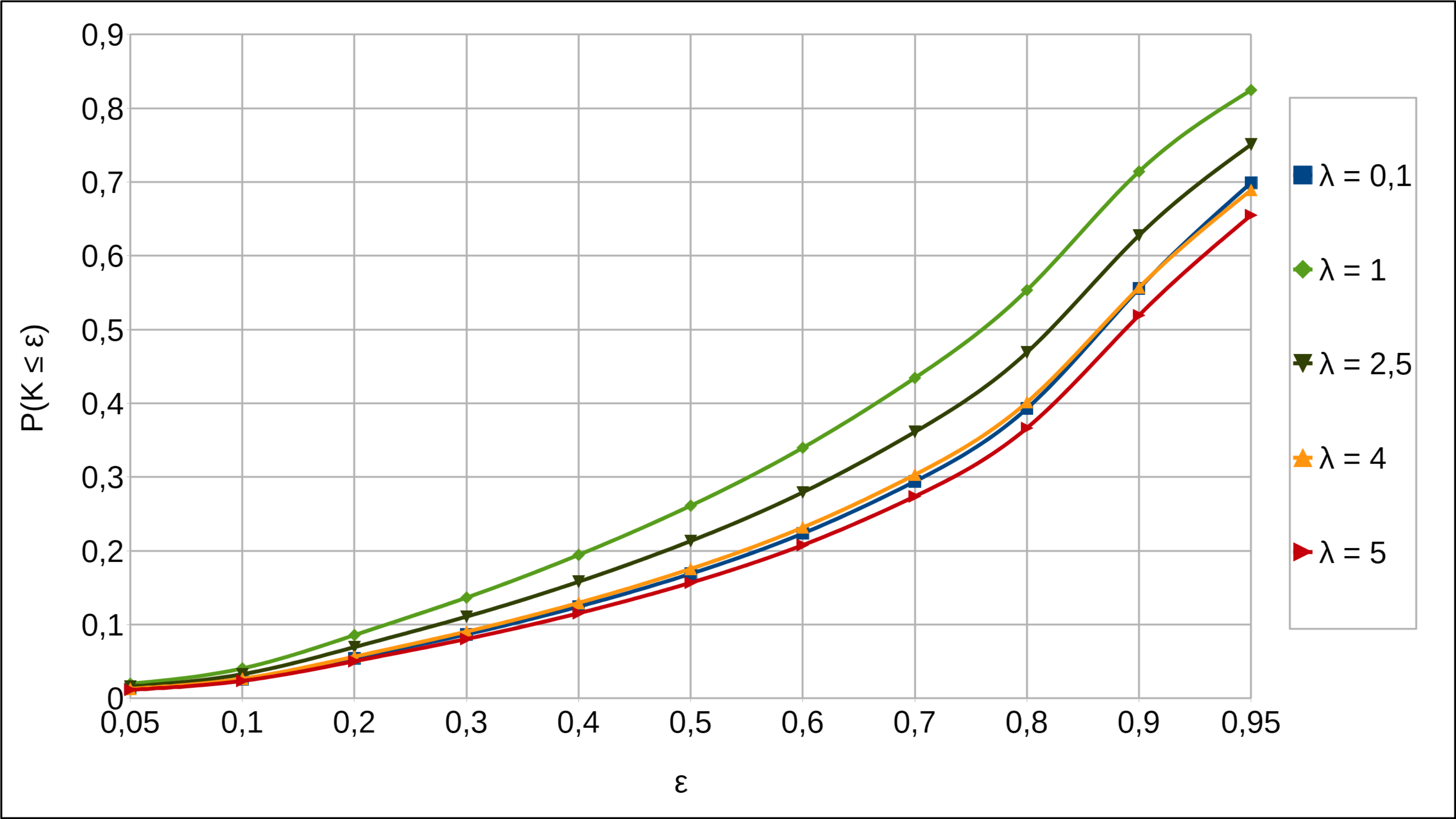}
  \caption{Exponential probabilities
  $P \left( K \leq \varepsilon \right) $ for different values of $\lambda$.}\label{FIG_LAGUERRE}
\end{figure}

\section{Normal perturbation of $3 \times 3$ matrices} \label{SEC_NORMAL3D}
Denote $\mathbb{R}_{+} ^{3} = \left\{ \left( x, z, y \right) :
x > 0, z > 0, y > 0 \right\}$ and consider the $\mathbb{R}_{+} ^{3}$-conditional normal measure of inconsistency
$P \left( K \leq \varepsilon \mid \mathbb{R}_{+} ^{3} \right)$ associated with the Lebesgue-Stieltjes measure 
$P = dF$ in $\mathbb{R} ^{3}$
generated by the density function of the form
\begin{equation*}
  f \left( t \right) = F ^{\prime} \left( t \right) 
    = \frac{1}{2} \frac{d}{dt} \Phi \left( 
    \frac{t - \mu}{\sigma \sqrt{2}} \right) 
    = \frac{1}{\sigma \sqrt{2 \pi}} \exp \left[ 
    -\frac{1}{2} \left( \frac{t - \mu}{\sigma} \right) ^{2} \right],
\end{equation*}
with the mean $\mu$ and the standard deviation $\sigma > 0$.
Without loss of generality, we shall assume below that $\mu = 0$. Under this assumption, we express the conditional probability
$P \left( K \leq \varepsilon \mid \mathbb{R}_{+} ^{3} \right)$ in terms of the arctan-exponential integral defined by
\begin{equation*}
  \Ti \left( x \right) = \int_{0} ^{\infty } \arctan 
    \left( tx \right) e ^{-t ^{2}} dt.
\end{equation*}
For this purpose, we employ the identity from \citep{Grad1}:
\begin{equation}
  \int_{0} ^{\infty} \Phi \left( t \right) e ^{-\alpha ^{2} 
    t ^{2}} dt =\frac{\sqrt{\pi}}{2 \alpha} 
    - \frac{\arctan \alpha}{\alpha \sqrt{\pi}}, \quad \RePart \alpha > 0. \label{EQ_ERRORFUNCINTEGRAL}
\end{equation}

\begin{theorem} \label{THR_NORMALPROBABILITY}
  If $0 < \varepsilon = 1 - \tau < 1$ and $\sigma > 0$, 
  then we have
  \begin{equation*}
    P \left( K \leq \varepsilon \mid \mathbb{R}_{+} ^{3}   
      \right) = \frac{4}{\pi \sqrt{\pi}} 
      \left[ \Ti \left( \frac{\sigma \sqrt{2}}{\tau} \right) 
           - \Ti \left( \tau \sigma \sqrt{2} \right) \right]. 
  \end{equation*}  
\end{theorem}

\begin{proof}
  The conditional probability is equal to
  \begin{equation*}
  P \left( K \leq 1 - \varepsilon \mid \mathbb{R}_{+} ^{3} 
      \right) = 8\int_{0}^{\infty }\int_{0}^{\infty }\left[ 
      F \left( \frac{xz}{\varepsilon} \right) 
      - F \left( \varepsilon xz \right) \right] dF 
      \left( z \right) dF \left( x \right).
  \end{equation*}
  Integrating by parts, we express this formula in the form
  \begin{align*}
   P \left( K \leq 1 - \varepsilon \mid \mathbb{R}_{+} ^{3} \right)
     & = -8 \int_{0} ^{\infty} \int_{0} ^{\infty}
     \left[ \frac{x}{\varepsilon} F ^{\prime} \left( 
     \frac{xz}{\varepsilon} \right) - \varepsilon x 
     F ^{\prime} \left( \varepsilon xz \right) \right] 
     F \left( z \right) dz d F \left( x \right)  \\
     & = \frac{4}{\sigma \sqrt{2\pi}} \int_{0} ^{\infty} \int_{0} ^{\infty} \Phi \left( \frac{z}{\sigma \sqrt{2}} \right) \left[ \varepsilon x e ^{-\left( 
     \frac{\varepsilon xz}{\sigma \sqrt{2}} \right) ^{2}}
     - \frac{x}{\varepsilon} e ^{-\left( 
     \frac{xz}{\sigma \sqrt{2}\varepsilon} \right) ^{2}}
     \right] dz d F\left( x \right).
  \end{align*}
  Now we can apply the identity \eqref{EQ_ERRORFUNCINTEGRAL} twice to get
  \begin{align*}
   P \left( K \leq 1 - \varepsilon \mid \mathbb{R}_{+} ^{3} \right)
    & = \frac{4}{\pi} \int_{0} ^{\infty} \left( \arctan 
    \frac{x}{\varepsilon} - \arctan \left( \varepsilon x \right)
    \right) F ^{\prime} \left( x \right) dx \\
    & = \frac{4}{\pi \sqrt{\pi}} \int_{0} ^{\infty}
    \left( \arctan \frac{\sigma \sqrt{2} x}{\varepsilon} 
    - \arctan \left( \varepsilon \sigma \sqrt{2} x \right) 
    \right) e ^{-x ^{2}} dx.
  \end{align*}
  Hence, the definition of the arctan-exponential integral yields the desired formula.
\end{proof}

Note that the probability in Theorem \ref{THR_NORMALPROBABILITY} can be reformulated into
\begin{equation}
  P \left( K \leq \varepsilon \mid \mathbb{R}_{+} ^{3}   
    \right) = \frac{2}{\pi \sqrt{\pi}}
    \int_{-\infty} ^{\infty} h \left( x \right) e ^{-x ^{2}} dx,
  \label{EQ_NORMAL_PROBABILITY}
\end{equation}
upon using the notation
\begin{equation*}
  h \left( x \right) = \left\lvert \arctan 
    \frac{\sigma \sqrt{2} x}{1 - \varepsilon} \right\rvert 
    - \left\lvert \arctan \left( \left( 1 - \varepsilon \right) \sigma \sqrt{2} x \right) \right\rvert.
\end{equation*}
Hence, it can be effectively approximated by the Gauss-Hermite quadrature formula
\begin{equation*}
  P \left( K \leq \varepsilon \mid \mathbb{R}_{+} ^{3}   
    \right) \approx \frac{2}{\pi \sqrt{\pi}}
    \sum_{k = 1} ^{n} w_{k} h \left( x_{k} \right),
\end{equation*}
where $w_{k}$ are weights and $x_{k}$ are nodes of the Gauss-Hermite quadrature of order $n$ \citep{Gaut1,Niki1}. As in the exponential case, these weights and nodes can be precomputed by using one of the popular algorithms such as the \cite{Gu1s} method or \cite{Glas1s} procedure. These calculations can be simplified by computing 
$P \left( K \leq \varepsilon \mid \mathbb{R}_{+} ^{3}   
\right)$ with the following MATLAB commands, derived from Theorem \ref{THR_NORMALPROBABILITY}:
\begin{verbatim}
  g = @(x, t) atan(t .* x) .* exp(-t .^ 2)
  Ti = @(x) integral(@(t) g(x, t), 0, Inf)
  P = @(eps, sig) (Ti(sig .* sqrt(2) ./ (1 - eps)) 
      - Ti((1 - eps) .* sig .* sqrt(2))) 
      .* 4 ./ (pi .* sqrt(pi))
\end{verbatim}
where the parameters \verb|eps| and \verb|sig| correspond to 
$\varepsilon$ and $\sigma$, respectively, and 
$\verb|P| \approx P \left( K \leq \varepsilon \mid 
\mathbb{R}_{+} ^{3} \right)$.
\medskip

\noindent {\bf Numerical experiments} \\
We used MATLAB to compute the probability values
$P \left( K \leq \varepsilon \mid 
\mathbb{R}_{+} ^{3} \right)$ for the parameters
\begin{equation*}
  \varepsilon = 0.05,\ 0.1,\ 0.2,\ 0.3,\ 0.4,\ 0,5,\ 0.6,\ 0.7,
  \ 0.9,\ 0.95
\end{equation*}
and
\begin{equation*}
  \sigma = 0.1,\ 0.2,\ 0.3,\ 0.4,\ 0.5,\ 0.6,\ 0.7,\ 0.8,
  \ 0.9,\ 1.0 .
\end{equation*}
The results we obtained were accurate to 4 decimal places.
They are illustrated in Figure \ref{FIG_HERMITE} as a three dimensional surface. In line with the previous sections,
we compared the generated probability values with the 
corresponding outputs of a Monte Carlo simulation in which
every random PC matrix $M$ had normally distributed entries $x, y, z$. The maximum discrepancy between associated 
probabilities remained within $0,0003$.

\begin{figure}[!ht]
  \centering \includegraphics[width=0.6\textwidth]{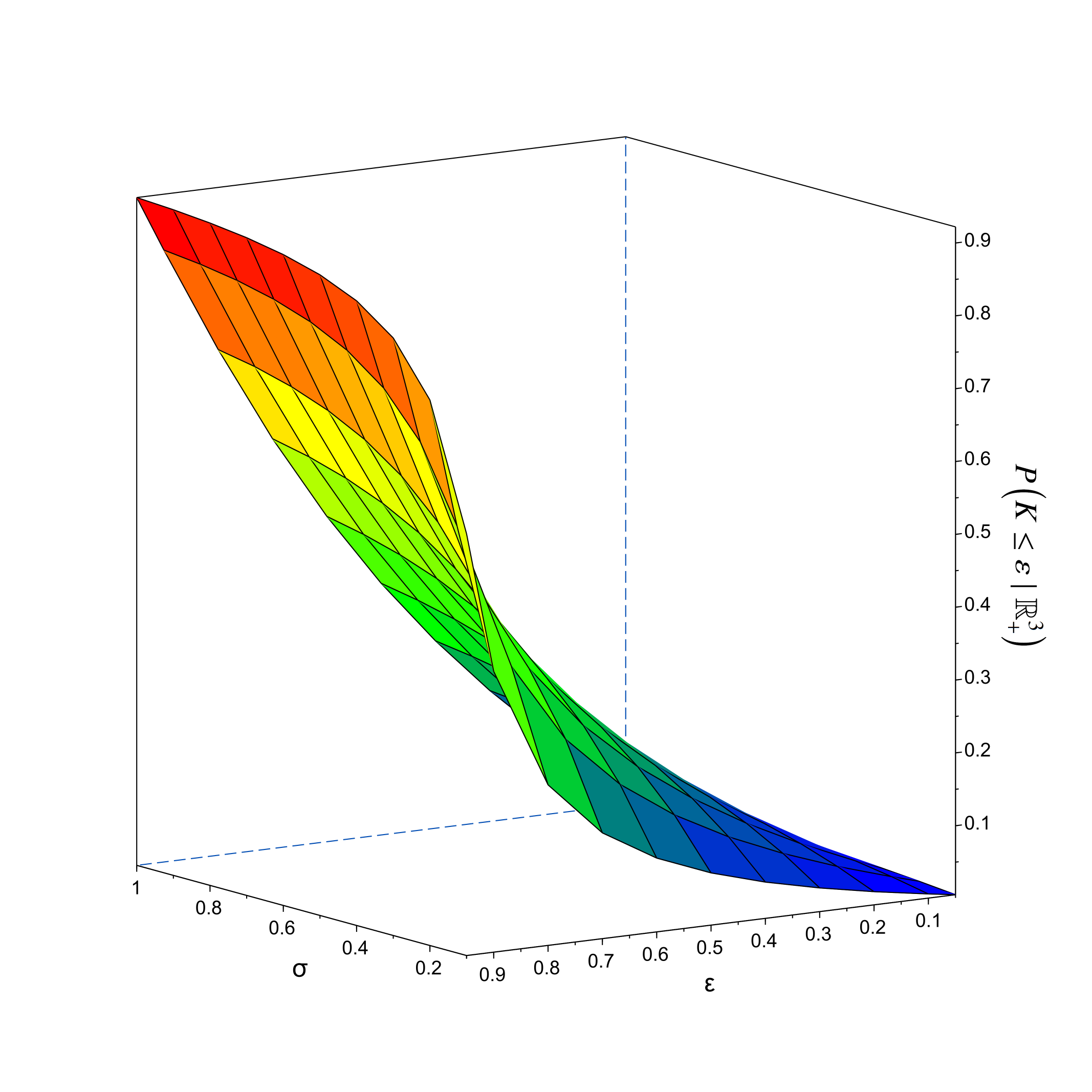}
  \caption{Normal probabilities $P \left( K \leq \varepsilon \mid \mathbb{R}_{+} ^{3} \right)$ for various values of $\varepsilon$ and $\sigma$.} \label{FIG_HERMITE}    
\end{figure}

\section{Perturbation of $n \times n$ matrices for $n > 3$} \label{SEC_PROBABILITYORDERn}
In this section, we consider the problem of computing probability
$P_{n} \left( K_{n} \leq \varepsilon \right)$ for any
Lebesgue-Stieltjes product measure $P_{n} = dF_{n}$ in
$\mathbb{R} ^{n}$, where $n \geq 3$. If $n = 3$, we denote
$P = P_{3}$. Then the indicator $K_{n} \left( M \right)$ of a
$n \times n$ PC matrix $M$ can be defined as in Section \ref{SEC_INTRO}:
\begin{equation*}
  K_{n} \left( M \right) = \max_{\left( i, j, k\right) \in
  \mathcal{T} \left( n \right)} K \left( M_{ijk} \right) ,
\end{equation*}
where $K$ is the Koczkodaj indicator
\begin{equation*}
  K \left( M_{ijk} \right) = 1 - \min \left\{ 
    \frac{ x_{j} }{ x_{i}x_{k} },
    \frac{ x_{i}x_{k} }{ x_{j} } \right\}
\end{equation*}
for $3 \times 3$ submatrices of the form
\begin{equation*}
    M_{ijk} = \left[
    \begin{array}{lll}
      1 & x_{i} & x_{j} \\
      1 / x_{i} & 1 & x_{k} \\
      1 / x_{j} & 1 / x_{k} & 1
    \end{array}  \right] .
\end{equation*}
However, this definition is disadvantageous for $n > 3$,
because $\binom{n}{3}$ random variables 
$M_{ijk} = M \left( x_{i}, x_{j}, x_{k} \right)$ are not
independent in general. From now on, we assume that the random
variables $M_{ijk}$ and $M_{pqr}$ are independent if the index
sets $\left\{ i, j, k \right\}$ and $\left\{ i, j, k \right\}$
are disjoint. A set of such variables is said to be independent
if every two of its elements are independent.

Unless otherwise stated, from now on we assume Hypothesis
\ref{HYP_NONOVERLAPPING} and choose a subset
$\mathcal{S} \left( n \right) \subset \mathcal{T} 
\left( n \right)$, described there. Furthermore, we suppose
that the definition of $K_{n} \left( M \right)$ is refined as
follows
\begin{equation*}
  K_{n} \left( M \right) = \max_{\left( i, j, k\right) \in
  \mathcal{S} \left( n \right)} K \left( M_{ijk} \right) .
\end{equation*}
Then the following important theorem follows from the
independence of random variables $M_{ijk}$, written under the
maximum.

\begin{theorem} \label{THR_PROBABILITYWITHSn}
  Under the assumption of Hypothesis \ref{HYP_NONOVERLAPPING},
  we have
  \begin{equation*}
    P_{n} \left( K_{n} \leq \varepsilon \right) = 
      \left[ P \left( K \leq \varepsilon \right) \right] 
      ^{\left\lvert \mathcal{S} \left( n \right) \right\rvert}
  \end{equation*}
  where 
  $\left\lvert \mathcal{S} \left( n \right) \right\rvert$
  is the number of elements in $\mathcal{S} \left( n \right)$.
\end{theorem}

\begin{proof}
  By Hypothesis \ref{HYP_NONOVERLAPPING} the matrix $M$ is a
  Cartesian product of disjoint submatrices 
  $M_{ijk} = M \left( x_{i}, x_{j}, x_{k} \right)$, where
  $\left( i, j, k \right) \in \mathcal{S} \left( n \right)$.
  Then we can apply Fubini’s theorem to show that the random
  variables $M_{ijk}$ are independent. Hence, the formula
  follows directly from the definition of independency.
\end{proof}

This theorem can be used to numerically compute
probabilities $P_{n} \left( K_{n} \leq \varepsilon \right)$ in
the case of uniform, exponential, and normal perturbations. For
this purpose, we recommend applying efficient procedures from Sections \ref{SEC_UNIFORM3D} - \ref{SEC_NORMAL3D} to compute the probability $P \left( M \leq \varepsilon \right)$.

Now we proceed to determine a subset of disjoint submatrices
$M_{ijk} = M \left( x_{i}, x_{j}, x_{k} \right)$ of $M$ having
maximal length
$\left\lvert \mathcal{S} \left( n \right) \right\rvert$. They
will be represented by a subset $\mathcal{S} \left( n \right)$
of the triad set $\mathcal{T} \left( n \right)$. In general,
such a set does not satisfy Hypothesis \ref{HYP_NONOVERLAPPING}.
Moreover, it is not uniquely determined. 

\begin{theorem} \label{THR_STRONGPROBABILITY}
  The set of triads $\mathcal{S} \left( n \right)$, $n \geq 3$,
  of length
  $\left\lvert \mathcal{S} \left( n \right) \right\rvert =
    \frac{n \left( n - 1 \right)}{6}$
  satisfies Hypothesis \ref{HYP_NONOVERLAPPING} for
  $n \equiv 1 \left(\text{mod } 6\right)$ or $n \equiv 3 \left(\text{mod } 6\right)$.
  Then we have
  \begin{equation*}
    P_{n} \left( K_{n} \leq \varepsilon \right) = 
      \left[ P \left( K \leq \varepsilon \right) \right] 
      ^{n \left( n - 1 \right) / 6}.
  \end{equation*}                                     
\end{theorem}

\begin{proof}
  The set $\mathcal{S} \left( n \right)$ is, in fact, a set of
  3-element subsets of the set
  $\mathcal{I} \left( n \right) = \left\{ 1, 2, \ldots, n
  \right\}$, such that each pair of distinct elements of 
  $\mathcal{I} \left( n \right)$ occurs together in exactly one triad of $\mathcal{S} \left( n \right)$. Therefore, the pair $\left( \mathcal{I} \left( n \right),
  \mathcal{S} \left( n \right) \right)$ is a combinatorial
  configuration called \emph{a Steiner triple system} of order $n$ \citep{Lind1}. It is known that such a
  configuration exists if and only if
  $n \equiv 1 \left(\text{mod } 6\right)$ or $n \equiv 3 \left(\text{mod } 6\right)$
  \citep{Kirk1s}. In view of the fact that each
  triad in $\mathcal{S} \left( n \right)$ contains
  three of $\binom{n}{2}$ pairs of distinct
  elements of $\mathcal{I} \left( n \right)$ and using Theorem \ref{THR_PROBABILITYWITHSn} we may complete the proof.  
\end{proof}

In Design Theory, there are two main techniques for constructing
Steiner triple systems \citep{Lind1}. The Bose Construction \citep{Bose1s} allows one to produce 
$\mathcal{S} \left( n \right)$ of order
$n \equiv 3 \left(\text{mod } 6 \right)$, whereas the Skolem
Construction \citep{Skol1s} covers order
$n \equiv 1 \left(\text{mod } 6 \right)$. These methods
are based on idempotent and half-idempotent commutative quasigroups, respectively. For example, for
$n = 3, 7, 9$, these procedures generate the following sets of
triads:
\begin{align*}
  \mathcal{S} \left( 3 \right) & = \left\{ 
    \left( 1, 2, 3 \right) \right\},\\  
  \mathcal{S} \left( 7 \right) & = \left\{ 
    \left( 1, 2, 3 \right), \left( 2, 4, 7 \right),
    \left( 3, 5, 7 \right), \left( 1, 6, 7 \right),
    \left( 1, 4, 5 \right), \left( 2, 5, 6 \right),
    \left( 3, 4, 6 \right) \right\},\\
  \mathcal{S} \left( 9 \right) & = \left\{ 
    \left( 1, 2, 3 \right), \left( 4, 5, 6 \right),
    \left( 7, 8, 9 \right), \left( 1, 4, 8 \right),
    \left( 2, 5, 9 \right), \left( 3, 6, 7 \right) \right.,\\
    & \quad\ 
    \left. \left( 1, 5, 7 \right), \left( 2, 6, 8 \right),
    \left( 3, 4, 9 \right), \left( 2, 4, 7 \right),
    \left( 3, 5, 8 \right), \left( 1, 6, 9 \right) \right\}
\end{align*}
The last two sets are shown in Table \ref{TAB_SFOR6AND7},
in which the triads are represented by squares numerated
according to the order within $\mathcal{S} \left( n \right)$.
\begin{table*}[h] 
  \caption{The sets $\mathcal{S} \left( n \right)$ of disjoint triads for $n = 7$ and $9$.} \label{TAB_SFOR6AND7}
  \begin{equation*}
    \begin{pmatrix}
      1 & \fcolorbox{black}{lightblue}{$1$} & \fcolorbox{black}{lightblue}{$1$} & \fcolorbox{black}{lightblue}{$5$}
        & \fcolorbox{black}{lightblue}{$5$} & \fcolorbox{black}{lightblue}{$4$} & \fcolorbox{black}{lightblue}{$4$} \\
      ~ & 1 & \fcolorbox{black}{lightblue}{$1$} & \fcolorbox{black}{lightblue}{$2$} & \fcolorbox{black}{lightblue}{$6$} & \fcolorbox{black}{lightblue}{$6$} & \fcolorbox{black}{lightblue}{$2$} \\
      ~ & ~ & 1 & \fcolorbox{black}{lightblue}{$7$} & \fcolorbox{black}{lightblue}{$3$} & \fcolorbox{black}{lightblue}{$7$} & \fcolorbox{black}{lightblue}{$3$} \\
      ~ & ~ & ~ & 1 & \fcolorbox{black}{lightblue}{$5$} & \fcolorbox{black}{lightblue}{$7$} & \fcolorbox{black}{lightblue}{$2$} \\
      ~ & ~ & ~ & ~ & 1 & \fcolorbox{black}{lightblue}{$6$} & \fcolorbox{black}{lightblue}{$3$} \\
      ~ & ~ & ~ & ~ & ~ & 1 & \fcolorbox{black}{lightblue}{$4$} \\
      ~ & ~ & ~ & ~ & ~ & ~ & 1
    \end{pmatrix}
    \quad
    \begin{pmatrix}
      1 & \fcolorbox{black}{lightblue}{$1$} & \fcolorbox{black}{lightblue}{$1$} & \fcolorbox{black}{lightblue}{$4$} & \fcolorbox{black}{lightblue}{$7$} 
        & \fcolorbox{black}{lightblue}{$12$} & \fcolorbox{black}{lightblue}{$7$} & \fcolorbox{black}{lightblue}{$4$} & \fcolorbox{black}{lightblue}{$12$} \\
      ~ & 1 & \fcolorbox{black}{lightblue}{$1$} & \fcolorbox{black}{lightblue}{$10$} & \fcolorbox{black}{lightblue}{$5$} & \fcolorbox{black}{lightblue}{$8$} 
        & \fcolorbox{black}{lightblue}{$10$} & \fcolorbox{black}{lightblue}{$8$} & \fcolorbox{black}{lightblue}{$5$} \\
      ~ & ~ & 1 & \fcolorbox{black}{lightblue}{$9$} & \fcolorbox{black}{lightblue}{$11$} & \fcolorbox{black}{lightblue}{$6$} & \fcolorbox{black}{lightblue}{$6$} & \fcolorbox{black}{lightblue}{$11$} & \fcolorbox{black}{lightblue}{$9$} \\
      ~ & ~ & ~ & 1 & \fcolorbox{black}{lightblue}{$2$} & \fcolorbox{black}{lightblue}{$2$} & \fcolorbox{black}{lightblue}{$10$} & \fcolorbox{black}{lightblue}{$4$} & \fcolorbox{black}{lightblue}{$9$} \\
      ~ & ~ & ~ & ~ & 1 & \fcolorbox{black}{lightblue}{$2$} & \fcolorbox{black}{lightblue}{$7$} & \fcolorbox{black}{lightblue}{$11$} & \fcolorbox{black}{lightblue}{$5$} \\
      ~ & ~ & ~ & ~ & ~ & 1 & \fcolorbox{black}{lightblue}{$6$} & \fcolorbox{black}{lightblue}{$8$} 
        & \fcolorbox{black}{lightblue}{$12$} \\
      ~ & ~ & ~ & ~ & ~ & ~ & 1 & \fcolorbox{black}{lightblue}{$3$} & \fcolorbox{black}{lightblue}{$3$} \\
      ~ & ~ & ~ & ~ & ~ & ~ & ~ & 1 & \fcolorbox{black}{lightblue}{$3$} \\
      ~ & ~ & ~ & ~ & ~ & ~ & ~ & ~ & 1
    \end{pmatrix}
  \end{equation*}   
\end{table*}

If $n \not\equiv 1 \left(\text{mod } 6\right)$ and
$n \not\equiv 3 \left(\text{mod } 6\right)$, then we denote
by $\mathcal{Q} \left( n \right)$ a subset of
$\mathcal{T} \left( n \right)$ that has the minimum length
$\left\lvert \mathcal{Q} \left( n \right) \right\rvert$
for which every entry $m_{rs}$ above the main diagonal of $M$
is an entry of at least one submatrix $M_{ijk}$ with
$\left( i, j, k \right) \in \mathcal{Q} \left( n \right)$.
In other words, the set
$\left\{ M_{ijk} : \left( i, j, k \right) 
\in \mathcal{Q} \left( n \right) \right\}$ is a cover of $M$
with the minimum possible length  
$\left\lvert \mathcal{Q} \left( n \right) \right\rvert$.

\begin{corollary}
  If $n \not\equiv 1 \left(\text{mod } 6\right)$ and 
  $n \not\equiv 3 \left(\text{mod } 6\right)$, then
  \begin{equation*}
    P_{n} \left( K_{n} \leq \varepsilon \right) < 
      \left[ P \left( K \leq \varepsilon \right) \right] 
      ^{\left\lvert \mathcal{Q} \left( n \right)    
        \right\rvert}.
  \end{equation*}
  where 
  \begin{equation*}
    \left\lvert \mathcal{Q} \left( n \right) \right\rvert =
    \left\lceil \frac{n}{3} \left\lceil \frac{n - 1}{2} \right\rceil \right\rceil
  \end{equation*}
  with symbol $\left\lceil \cdot \right\rceil$ denoting the
  ceil function.
\end{corollary}

\begin{proof}
  The problem of finding $\mathcal{Q} \left( n \right)$
  may be treated as a minimum covering with triangles. 
  Since in \citep{Scho3s} the number 
  $\psi \left( 3, 2, 1, n \right)$ of triangles in this minimum
  covering has been established in an even more general setting
  \begin{equation*}
    \psi \left( k, l, \lambda, n \right) = 
      \left\lceil \frac{n}{k} \left\lceil 
        \frac{n - 1}{k - 1} \left\lceil \cdots \left\lceil
        \frac{n - l + 2}{k - l + 2} \left\lceil
        \frac{\lambda \left( n - l + 1 \right)}{k - l + 1}
        \right\rceil \right\rceil \cdots \right\rceil
      \right\rceil \right\rceil,
  \end{equation*}
  we may finish the proof.
\end{proof}

A more specific formula for
$\left\lvert \mathcal{Q} \left( n \right) \right\rvert$ is known, which was proved earlier in \citep{Fort1s}.
\begin{equation*}
    \left\lvert \mathcal{Q} \left( n \right) \right\rvert = \left\{ 
    \begin{array}{ll}
      n^{2} / 6, & n \equiv 0 \left(\text{mod } 6\right), \\
      \left( n^{2} + 2 \right) / 6, 
        & n \equiv 2 \text{ or } n \equiv 4 
        \left(\text{mod } 6\right), \\
      \left( n^{2} - n + 4 \right) / 6,
        & n \equiv 5 \left(\text{mod } 6\right).
    \end{array}
  \right.
\end{equation*}
For example, we have
\begin{equation*}
  \left\lvert \mathcal{Q} \left( 4 \right) \right\rvert 
    = 3,\quad
  \left\lvert \mathcal{Q} \left( 5 \right) \right\rvert 
    = 4,\quad
  \left\lvert \mathcal{Q} \left( 6 \right) \right\rvert
    = 6,\quad \text{and}\quad
  \left\lvert \mathcal{Q} \left( 8 \right) \right\rvert
    = 11.
\end{equation*}
The procedures for constructing minimum covering with
triangles can be found in \citep{Lind1}. For example, for
$n = 4, 5, 6$ and $8$, they can be used to produce the following sets of triads (see also Table \ref{TAB_QFOR8}):
\begin{align*}
  \mathcal{Q} \left( 4 \right) & = \left\{ 
    \left( 1, 2, 3 \right), \left( 1, 2, 4 \right),
    \left( 1, 3, 4 \right) \right\},\\  
  \mathcal{Q} \left( 5 \right) & = \left\{ 
    \left( 1, 2, 3 \right), \left( 1, 2, 4 \right),
    \left( 1, 2, 5 \right), \left( 3, 4, 5 \right) \right\},\\
  \mathcal{Q} \left( 6 \right) & = \left\{ 
    \left( 1, 2, 3 \right), \left( 1, 2, 4 \right),
    \left( 1, 5, 6 \right), \left( 2, 5, 6 \right),
    \left( 3, 4, 5 \right), \left( 3, 4, 6 \right)
    \right\},\\    
  \mathcal{Q} \left( 8 \right) & = \left\{ 
    \left( 1, 2, 3 \right), \left( 1, 2, 7 \right),
    \left( 1, 3, 8 \right), \left( 1, 4, 5 \right),
    \left( 1, 4, 6 \right), \left( 2, 4, 8 \right) \right.,\\
    & \quad\ 
    \left. \left( 2, 5, 6 \right), \left( 3, 4, 7 \right),
    \left( 3, 5, 6 \right), \left( 5, 7, 8 \right),
    \left( 6, 7, 8 \right) \right\}.
\end{align*}

\begin{table*}[ht] 
  \caption{The minimal triangle cover 
  $\mathcal{Q} \left( n \right)$ of 
  $\mathcal{T} \left( n \right)$ for $n = 8$.} \label{TAB_QFOR8}
  \begin{equation*}    
    \begin{pmatrix}
      1 & \fcolorbox{black}{lightblue}{$1 \mid 2$} & \fcolorbox{black}{lightblue}{$1 \mid 3$} & \fcolorbox{black}{lightblue}{$4 \mid 5$} & \fcolorbox{black}{lightblue}{$4$} 
        & \fcolorbox{black}{lightblue}{$5$} & \fcolorbox{black}{lightblue}{$2$} & \fcolorbox{black}{lightblue}{$3$} \\
      ~ & 1 & \fcolorbox{black}{lightblue}{$1$} & \fcolorbox{black}{lightblue}{$6$} & \fcolorbox{black}{lightblue}{$7$} & \fcolorbox{black}{lightblue}{$7$} 
        & \fcolorbox{black}{lightblue}{$2$} & \fcolorbox{black}{lightblue}{$6$} \\
      ~ & ~ & 1 & \fcolorbox{black}{lightblue}{$8$} & \fcolorbox{black}{lightblue}{$9$} & \fcolorbox{black}{lightblue}{$9$} & \fcolorbox{black}{lightblue}{$8$} & \fcolorbox{black}{lightblue}{$3$} \\
      ~ & ~ & ~ & 1 & \fcolorbox{black}{lightblue}{$4$} & \fcolorbox{black}{lightblue}{$5$} & \fcolorbox{black}{lightblue}{$8$} & \fcolorbox{black}{lightblue}{$6$} \\
      ~ & ~ & ~ & ~ & 1 & \fcolorbox{black}{lightblue}{$7 \mid 9$} & \fcolorbox{black}{lightblue}{$10$} & \fcolorbox{black}{lightblue}{$10$} \\
      ~ & ~ & ~ & ~ & ~ & 1 & \fcolorbox{black}{lightblue}{$11$} & \fcolorbox{black}{lightblue}{$11$} \\
      ~ & ~ & ~ & ~ & ~ & ~ & 1 & \fcolorbox{black}{lightblue}{$10 \mid 11$} \\
      ~ & ~ & ~ & ~ & ~ & ~ & ~ & 1 
    \end{pmatrix}
  \end{equation*}   
\end{table*}

\section{Conclusions}
The presented formulae for probabilities 
$P_{n} \left( K_{n} \leq \varepsilon \right)$ allow one to estimate the frequency of the occurrence of random matrices within the assumed level of inconsistency relative to all admissible PC matrices. They suggest a novel way to compute a threshold value of the inconsistency indicator 
$K_{n} \left( M \right)$, below which one should accept the PC
matrices assessed by experts. More precisely, if we suppose that $m$ experts made their judgements in the form of PC matrices $M ^{ \left( k \right) }$, $k = 1, \ldots, m$, then every value of $K_{n} \left( M ^{\left( k \right)} \right) $ must be compared with the assumed inconsistency tolerance threshold $\varepsilon \in \left( 0,1 \right) $. Only those PC matrices for which $K_{n} \left( M ^{\left( k \right)} \right) \leq \varepsilon $ are allowed to the further stages of the decision making process. If we choose some $\varepsilon$ as the threshold of acceptable inconsistency of experts' assessments, then we admit that in this group of experts there may be 
\begin{equation*}
  r \leq m\, P_{n} \left( K_{n} \leq \varepsilon \right) \ \text{persons},
\end{equation*}
who could provide random PC matrices instead of meaningful ones. Thus, it is obvious that a group decision making system should adjust the level $\varepsilon$ of acceptable inconsistency of PC matrices to the largest possible value for which $r < 1$. For this purpose, the following divide-and-conquer Algorithm \ref{ALG_ADJUSTINGTHRESHOLD} can be used.

\begin{algorithm}
  \caption{\small Adjusting the threshold of PC matrix inconsistency, according to the number of experts.} \label{ALG_ADJUSTINGTHRESHOLD}
  \begin{algorithmic}[1]
  {\small
    \REQUIRE an initial level of inconsistency $\varepsilon _{0} \in \left( 0,1\right) $, a number of experts $m \in \mathbb{N}$, and required precision of computation $prec > 0$.
    \ENSURE the altered level of inconsitency $\varepsilon$.
    \IF{$\lfloor m P_{n} \left( K_{n} \leq \varepsilon_{0} \right) \rfloor < 1$}
      \RETURN $\varepsilon_{0}$
    \ENDIF
    \STATE $l \leftarrow 0$
    \STATE $r \leftarrow 1$
    \STATE $\varepsilon \leftarrow \varepsilon_{0}$
    \WHILE{$r - l > prec$}
      \IF{$\lfloor m P_{n} \left( K_{n} \leq \varepsilon \right) \rfloor \geq 1$}
        \STATE $r \leftarrow \varepsilon$
      \ELSE
	\STATE $l \leftarrow \varepsilon$
      \ENDIF
      \STATE $\varepsilon \leftarrow \left( l + r \right) / 2$
    \ENDWHILE \\
    \RETURN  $l$
  }
  \end{algorithmic}
\end{algorithm}

Finally, since any cumulative probability distribution function is non-decreasing in each direction, the correctness of Algorithm \ref{ALG_ADJUSTINGTHRESHOLD} is evident. Moreover, note that the applicability of this algorithm is indirectly determined by the degree $n$ of PC matrices. The larger $n$, the more experts are required for this threshold adjustment to work. Fortunately, in view of psychological limitations of human insight, in practice the number $n$ of alternatives included in pairwise comparisons is usually assumed to be at most $7$. If, however, a greater number of alternatives have to be considered, then they may be divided into smaller groups before pairwise comparisons are made. 

\bibliographystyle{model5-names}
\biboptions{authoryear}
\bibliography{IiAHPUDU}

\begin{thebibliography}{32}
\expandafter\ifx\csname natexlab\endcsname\relax\def\natexlab#1{#1}\fi
\providecommand{\url}[1]{\texttt{#1}}
\providecommand{\href}[2]{#2}
\providecommand{\path}[1]{#1}
\providecommand{\DOIprefix}{doi:}
\providecommand{\ArXivprefix}{arXiv:}
\providecommand{\URLprefix}{URL: }
\providecommand{\Pubmedprefix}{pmid:}
\providecommand{\doi}[1]{\href{http://dx.doi.org/#1}{\path{#1}}}
\providecommand{\Pubmed}[1]{\href{pmid:#1}{\path{#1}}}
\providecommand{\bibinfo}[2]{#2}
\ifx\xfnm\relax \def\xfnm[#1]{\unskip,\space#1}\fi
\bibitem[{Aguar\'on \& Moreno-Jim\'enez(2003)}]{Agua1s}
\bibinfo{author}{Aguar\'on, J.}, \& \bibinfo{author}{Moreno-Jim\'enez, J.~M.}
  (\bibinfo{year}{2003}).
\newblock \bibinfo{title}{The geometric consistency index: {Approximated}
  thresholds}.
\newblock {\it \bibinfo{journal}{European J. Oper. Res.}\/},  {\it
  \bibinfo{volume}{147}\/}, \bibinfo{pages}{137--145}.
\bibitem[{Bose(1939)}]{Bose1s}
\bibinfo{author}{Bose, R.~C.} (\bibinfo{year}{1939}).
\newblock \bibinfo{title}{On the construction of balanced incomplete block
  designs}.
\newblock {\it \bibinfo{journal}{Ann. Eugenics}\/},  {\it
  \bibinfo{volume}{9}\/}, \bibinfo{pages}{353–399}.
\bibitem[{Brunelli(2016)}]{Brun1s}
\bibinfo{author}{Brunelli, M.} (\bibinfo{year}{2016}).
\newblock \bibinfo{title}{Recent advances on inconsistency indices for pairwise
  comparisons -- a commentary}.
\newblock {\it \bibinfo{journal}{Fundam. Inform.}\/},  {\it
  \bibinfo{volume}{144}\/}, \bibinfo{pages}{321--332}.
\bibitem[{Brunelli(2018)}]{Brun3s}
\bibinfo{author}{Brunelli, M.} (\bibinfo{year}{2018}).
\newblock \bibinfo{title}{A survey of inconsistency indices for pairwise
  comparisons}.
\newblock {\it \bibinfo{journal}{Int. J. Gen. Syst.}\/},  {\it
  \bibinfo{volume}{47}\/}, \bibinfo{pages}{751--771}.
\bibitem[{Brunelli et~al.(2013)Brunelli, Canal \& Fedrizzi}]{Brun4s}
\bibinfo{author}{Brunelli, M.}, \bibinfo{author}{Canal, L.}, \&
  \bibinfo{author}{Fedrizzi, M.} (\bibinfo{year}{2013}).
\newblock \bibinfo{title}{Inconsistency indices for pairwise comparison
  matrices: a numerical study}.
\newblock {\it \bibinfo{journal}{Ann. Oper. Res.}\/},  {\it
  \bibinfo{volume}{211}\/}, \bibinfo{pages}{493--509}.
\bibitem[{Duszak \& Koczkodaj(1994)}]{Dusz1s}
\bibinfo{author}{Duszak, Z.}, \& \bibinfo{author}{Koczkodaj, W.~W.}
  (\bibinfo{year}{1994}).
\newblock \bibinfo{title}{Generalization of a new definition of consistency for
  pairwise comparisons}.
\newblock {\it \bibinfo{journal}{Inf. Process. Lett.}\/},  {\it
  \bibinfo{volume}{52}\/}, \bibinfo{pages}{273--276}.
\bibitem[{{Fort Jr.} \& Hedlund(1958)}]{Fort1s}
\bibinfo{author}{{Fort Jr.}, M.~K.}, \& \bibinfo{author}{Hedlund, G.~A.}
  (\bibinfo{year}{1958}).
\newblock \bibinfo{title}{Minimal coverings of pairs by triples}.
\newblock {\it \bibinfo{journal}{Pacific J. Math.}\/},  {\it
  \bibinfo{volume}{8}\/}, \bibinfo{pages}{709--719}.
\bibitem[{Gautschi(2004)}]{Gaut1}
\bibinfo{author}{Gautschi, W.} (\bibinfo{year}{2004}).
\newblock {\it \bibinfo{title}{{Orthogonal Polynomials Computation and
  Approximation}}\/}.
\newblock Numerical Mathematics and Scientific Computation.
\newblock \bibinfo{address}{New York}: \bibinfo{publisher}{Oxford University
  Press}.
\bibitem[{Glaser et~al.(2007)Glaser, Liu \& Rokhlin}]{Glas1s}
\bibinfo{author}{Glaser, A.}, \bibinfo{author}{Liu, X.}, \&
  \bibinfo{author}{Rokhlin, V.} (\bibinfo{year}{2007}).
\newblock \bibinfo{title}{A fast algorithm for the calculation of the roots of
  special functions}.
\newblock {\it \bibinfo{journal}{SIAM J. Sci. Comput.}\/},  {\it
  \bibinfo{volume}{29}\/}, \bibinfo{pages}{1420--1438}.
\bibitem[{Gradshteyn \& Ryzhik(1994)}]{Grad1}
\bibinfo{author}{Gradshteyn, I.~S.}, \& \bibinfo{author}{Ryzhik, I.~M.}
  (\bibinfo{year}{1994}).
\newblock {\it \bibinfo{title}{{Table of Integrals, Series, and Products}}\/}.
\newblock (\bibinfo{edition}{5th} ed.).
\newblock \bibinfo{address}{San Diego}: \bibinfo{publisher}{Academic Press}.
\bibitem[{Gu \& Eisenstat(1995)}]{Gu1s}
\bibinfo{author}{Gu, M.}, \& \bibinfo{author}{Eisenstat, S.~C.}
  (\bibinfo{year}{1995}).
\newblock \bibinfo{title}{A divide-and-conquer algorithm for the symmetric
  tridiagonal eigenproblem}.
\newblock {\it \bibinfo{journal}{SIAM J. Matrix Anal. Appl.}\/},  {\it
  \bibinfo{volume}{16}\/}, \bibinfo{pages}{172--191}.
\bibitem[{Herman \& Koczkodaj(1996)}]{Herm1s}
\bibinfo{author}{Herman, M.~W.}, \& \bibinfo{author}{Koczkodaj, W.~W.}
  (\bibinfo{year}{1996}).
\newblock \bibinfo{title}{A {Monte Carlo} study of pairwise comparison}.
\newblock {\it \bibinfo{journal}{Inform. Process. Lett.}\/},  {\it
  \bibinfo{volume}{57}\/}, \bibinfo{pages}{25--29}.
\bibitem[{Kendall \& Babington-Smith(1940)}]{Kend1}
\bibinfo{author}{Kendall, M.~G.}, \& \bibinfo{author}{Babington-Smith, B.}
  (\bibinfo{year}{1940}).
\newblock \bibinfo{title}{On the method of paired comparisons}.
\newblock {\it \bibinfo{journal}{Biometrika}\/},  {\it \bibinfo{volume}{31}\/},
  \bibinfo{pages}{324--345}.
\bibitem[{Kirkman(1847)}]{Kirk1s}
\bibinfo{author}{Kirkman, T.~P.} (\bibinfo{year}{1847}).
\newblock \bibinfo{title}{On a problem in combinations}.
\newblock {\it \bibinfo{journal}{Cambridge and Dublin Math. Journal}\/},  {\it
  \bibinfo{volume}{2}\/}, \bibinfo{pages}{191–204}.
\bibitem[{Koczkodaj(1993)}]{Kocz4s}
\bibinfo{author}{Koczkodaj, W.~W.} (\bibinfo{year}{1993}).
\newblock \bibinfo{title}{A new definition of consistency of pairwise
  comparisons}.
\newblock {\it \bibinfo{journal}{Math. Comput. Model.}\/},  {\it
  \bibinfo{volume}{18}\/}, \bibinfo{pages}{79--84}.
\bibitem[{Koczkodaj et~al.(2020)Koczkodaj, Smarzewski \& Szybowski}]{Kocz2s}
\bibinfo{author}{Koczkodaj, W.~W.}, \bibinfo{author}{Smarzewski, R.}, \&
  \bibinfo{author}{Szybowski, J.} (\bibinfo{year}{2020}).
\newblock \bibinfo{title}{On orthogonal projections on the space of consistent
  pairwise comparisons matrices}.
\newblock {\it \bibinfo{journal}{Fundam. Inform.}\/},  {\it
  \bibinfo{volume}{172}\/}, \bibinfo{pages}{379--397}.
\bibitem[{Koczkodaj \& Szwarc(2014)}]{Kocz5s}
\bibinfo{author}{Koczkodaj, W.~W.}, \& \bibinfo{author}{Szwarc, R.}
  (\bibinfo{year}{2014}).
\newblock \bibinfo{title}{On axiomatization of inconsistency indicators for
  pairwise comparisons}.
\newblock {\it \bibinfo{journal}{Fundam. Inform.}\/},  {\it
  \bibinfo{volume}{132}\/}, \bibinfo{pages}{485--500}.
\bibitem[{Kułakowski et~al.(2019)Kułakowski, Szybowski \& Prusak}]{Kula1s}
\bibinfo{author}{Kułakowski, K.}, \bibinfo{author}{Szybowski, J.}, \&
  \bibinfo{author}{Prusak, A.} (\bibinfo{year}{2019}).
\newblock \bibinfo{title}{Towards quantification of incompleteness in the
  pairwise comparisons methods}.
\newblock {\it \bibinfo{journal}{Internat. J. Approx. Reason.}\/},  {\it
  \bibinfo{volume}{115}\/}, \bibinfo{pages}{221--234}.
\bibitem[{Lindner \& Rodger(2008)}]{Lind1}
\bibinfo{author}{Lindner, C.~C.}, \& \bibinfo{author}{Rodger, C.~A.}
  (\bibinfo{year}{2008}).
\newblock {\it \bibinfo{title}{{Design Theory}}\/}.
\newblock Discrete Mathematics and its Applications.
\newblock \bibinfo{address}{Boca Raton, FL}: \bibinfo{publisher}{CRC Press,
  Taylor \& Francis Group}.
\bibitem[{Liu et~al.(2020)Liu, Zou \& Li}]{Liu1s}
\bibinfo{author}{Liu, F.}, \bibinfo{author}{Zou, S.-C.}, \&
  \bibinfo{author}{Li, Q.} (\bibinfo{year}{2020}).
\newblock \bibinfo{title}{Deriving priorities from pairwise comparison matrices
  with a novel consistency index}.
\newblock {\it \bibinfo{journal}{Appl. Math. Comput.}\/},  {\it
  \bibinfo{volume}{374}\/}, \bibinfo{pages}{125059}.
\bibitem[{Magnot(2019)}]{Magn1}
\bibinfo{author}{Magnot, J.-P.} (\bibinfo{year}{2019}).
\newblock \bibinfo{title}{On mathematical structures on pairwise comparisons
  matrices with coefficients in an abstract group arising from quantum
  gravity}.
\newblock {\it \bibinfo{journal}{Heliyon}\/},  {\it \bibinfo{volume}{5}\/},
  \bibinfo{pages}{e01821}.
\bibitem[{Magnot(2020)}]{Magn2}
\bibinfo{author}{Magnot, J.-P.} (\bibinfo{year}{2020}).
\newblock \bibinfo{title}{Inconsistency in pairwise comparisons as an {Abelian}
  {Yang-Mills} theory}.
\newblock \bibinfo{howpublished}{hal-03091527
  (https://hal.archives-ouvertes.fr/hal-03091527/document)}.
\bibitem[{Mazurek(2023)}]{Mazu1}
\bibinfo{author}{Mazurek, J.} (\bibinfo{year}{2023}).
\newblock \bibinfo{title}{{Inconsistency Indices and Their Properties}}.
\newblock In {\it \bibinfo{booktitle}{Advances in Pairwise Comparisons:
  Detection, Evaluation and Reduction of Inconsistency}\/} (pp.
  \bibinfo{pages}{21--47}).
\newblock \bibinfo{address}{Cham}: \bibinfo{publisher}{Springer Nature
  Switzerland}.
\bibitem[{Nikiforov \& Uvarov(1988)}]{Niki1}
\bibinfo{author}{Nikiforov, A.~F.}, \& \bibinfo{author}{Uvarov, V.~B.}
  (\bibinfo{year}{1988}).
\newblock {\it \bibinfo{title}{Special Functions of Mathematical Physics}\/}.
\newblock \bibinfo{address}{Basel-Boston}:
  \bibinfo{publisher}{{Birk\-h\"{a}user}}.
\bibitem[{Pel\'aez \& Lamata(2003)}]{Pela1}
\bibinfo{author}{Pel\'aez, J.~I.}, \& \bibinfo{author}{Lamata, M.~T.}
  (\bibinfo{year}{2003}).
\newblock \bibinfo{title}{A new measure of consistency for positive reciprocal
  matrices}.
\newblock {\it \bibinfo{journal}{Computers and Mathematics with
  Applications}\/},  {\it \bibinfo{volume}{46}\/}, \bibinfo{pages}{1839--1845}.
\bibitem[{Saaty(1977)}]{Saat1s}
\bibinfo{author}{Saaty, T.~L.} (\bibinfo{year}{1977}).
\newblock \bibinfo{title}{A scaling method for priorities in hierarchical
  structure}.
\newblock {\it \bibinfo{journal}{J. Math. Psych.}\/},  {\it
  \bibinfo{volume}{15}\/}, \bibinfo{pages}{234--281}.
\bibitem[{Sch{\"o}nheim(1964)}]{Scho3s}
\bibinfo{author}{Sch{\"o}nheim, J.} (\bibinfo{year}{1964}).
\newblock \bibinfo{title}{On coverings}.
\newblock {\it \bibinfo{journal}{Pacific J. Math.}\/},  {\it
  \bibinfo{volume}{14}\/}, \bibinfo{pages}{1405–1411}.
\bibitem[{Skolem(1958)}]{Skol1s}
\bibinfo{author}{Skolem, T.} (\bibinfo{year}{1958}).
\newblock \bibinfo{title}{Some remarks on the triple systems of steiner}.
\newblock {\it \bibinfo{journal}{Math. Scand.}\/},  {\it
  \bibinfo{volume}{6}\/}, \bibinfo{pages}{273–280}.
\bibitem[{Smarzewski \& Rutka(2020)}]{Smar3s}
\bibinfo{author}{Smarzewski, R.}, \& \bibinfo{author}{Rutka, P.}
  (\bibinfo{year}{2020}).
\newblock \bibinfo{title}{Consistent projections and indicators in pairwise
  comparisons}.
\newblock {\it \bibinfo{journal}{Internat. J. Approx. Reason.}\/},  {\it
  \bibinfo{volume}{124}\/}, \bibinfo{pages}{123--132}.
\bibitem[{Stein \& Mizzi(2007)}]{Stei1s}
\bibinfo{author}{Stein, W.~E.}, \& \bibinfo{author}{Mizzi, P.~J.}
  (\bibinfo{year}{2007}).
\newblock \bibinfo{title}{The harmonic consistency index for the analytic
  hierarchy process}.
\newblock {\it \bibinfo{journal}{European J. Oper. Res.}\/},  {\it
  \bibinfo{volume}{177}\/}, \bibinfo{pages}{488--497}.
\bibitem[{Szybowski et~al.(2020)Szybowski, Ku\l{}akowski \& Prusak}]{Szyb1s}
\bibinfo{author}{Szybowski, J.}, \bibinfo{author}{Ku\l{}akowski, K.}, \&
  \bibinfo{author}{Prusak, A.} (\bibinfo{year}{2020}).
\newblock \bibinfo{title}{New inconsistency indicators for incomplete pairwise
  comparisons matrices}.
\newblock {\it \bibinfo{journal}{Math. Social Sci.}\/},  {\it
  \bibinfo{volume}{108}\/}, \bibinfo{pages}{138--145}.
\bibitem[{{Van Laarhoven} \& Pedrycz(1983)}]{Laar1s}
\bibinfo{author}{{Van Laarhoven}, P.~J.~M.}, \& \bibinfo{author}{Pedrycz, W.}
  (\bibinfo{year}{1983}).
\newblock \bibinfo{title}{A fuzzy extension of {Saaty's} priority theory}.
\newblock {\it \bibinfo{journal}{Fuzzy Sets and Systems}\/},  {\it
  \bibinfo{volume}{11}\/}, \bibinfo{pages}{229--241}.

\end{thebibliography}

\end{document}